\documentclass[11pt]{article}
\pdfoutput=1
\usepackage{fullpage}
\usepackage{amsmath}
\usepackage{amssymb}
\newenvironment{proof}[1][Proof.]{\begin{trivlist}
\item[\hskip \labelsep {\bfseries #1}]}{\hfill$\blacksquare$\end{trivlist}}

\usepackage{graphicx}
\usepackage{color}
\usepackage{alltt}
\usepackage{algorithm}
\usepackage{algorithmic}
\usepackage{multicol}
\usepackage{tabularx}

\newtheorem{theorem}{Theorem}

\newcommand{\leftchildarg}[1]{\text{left-child}(\ensuremath{#1})}
\newcommand{\rightchildarg}[1]{\text{right-child}(\ensuremath{#1})}
\newcommand{\parentarg}[1]{\text{parent}(\ensuremath{#1})}
\newcommand{\childarg}[2]{\text{child}(\ensuremath{#1,#2})}
\newcommand{\subtreesizearg}[1]{\text{subtree-size}(\ensuremath{#1})}
\newcommand{\lcaarg}[2]{\text{LCA}(\ensuremath{#1,#2})}
\newcommand{\levelancarg}[2]{\text{level-ancestor}(\ensuremath{#1,#2})}

\newcommand{\degreearg}[1]{\ensuremath{\text{degree}(#1)}}
\newcommand{\deptharg}[1]{\ensuremath{\text{depth}(#1)}}

\newcommand{\lmleafarg}[1]{\ensuremath{\text{leftmost-leaf}(#1)}}
\newcommand{\rmleafarg}[1]{\ensuremath{\text{rightmost-leaf}(#1)}}

\newcommand{\nextsiblingarg}[1]{\ensuremath{\text{next-sibling}(#1)}}
\newcommand{\prevsiblingarg}[1]{\ensuremath{\text{previous-sibling}(#1)}}

\newcommand{\leftchild}{\text{left-child}}
\newcommand{\rightchild}{\text{right-child}}
\newcommand{\parent}{\text{parent}}
\newcommand{\child}{\text{child}}
\newcommand{\subtreesize}{\text{subtree-size}}
\newcommand{\lca}{\text{LCA}}

\newcommand{\depth}{\text{depth}}
\newcommand{\levelanc}{\text{level-ancestor}}

\newcommand{\lmleaf}{\text{leftmost-leaf}}
\newcommand{\rmleaf}{\text{rightmost-leaf}}

\newcommand{\select}{\text{select}}

\newcommand{\preorder}{\text{preorder}}
\newcommand{\postorder}{\text{postorder}}
\newcommand{\inorder}{\text{inorder}}
\newcommand{\preorderRight}{\text{preorder-right}}
\newcommand{\postorderRight}{\text{postorder-right}}
\newcommand{\dfuds}{\text{DFUDS}}
\newcommand{\dfudsRight}{\text{DFUDS-RtL}}

\mathchardef\mhyphen="2D

\newcommand{\rmqarg}[2]{\ensuremath{\text{RMQ}(#1,#2)}}

\newcommand{\preorderarg}[1]{\ensuremath{\text{\preorder}(#1)}}
\newcommand{\preorderRightarg}[1]{\ensuremath{\text{\preorderRight}(#1)}}
\newcommand{\inorderarg}[1]{\ensuremath{\text{\inorder}(#1)}}
\newcommand{\postorderarg}[1]{\ensuremath{\text{\postorder}(#1)}}
\newcommand{\postorderRightarg}[1]{\ensuremath{\text{\postorderRight}(#1)}}

\newcommand{\trans}[1]{$\tau_{#1}$}
\newcommand{\transarg}[2]{\ensuremath{\tau_{#1}(#2)}}
\newcommand{\transTree}[1]{$T_{#1}$}
\newcommand{\transOrdering}[3]{$\tau_{#1}(\text{#2}) = \text{#3}$}
\newcommand{\bintree}{$T_b$}

\newcommand{\zak}{Zaks' sequence}

\newcommand{\minheap}{2d-Min-Heap}

\begin{document}

\title{On Succinct Representations of Binary Trees\thanks{An abstract of some of the results in this paper appeared in %
\emph{Computing and Combinatorics: Proceedings of the 18th Annual
International Conference COCOON 2012}, Springer LNCS 7434, pp. 396--407, 2012.}}

\author{Pooya Davoodi\\New York University, Polytechnic School of Engineering\\\texttt{pooyadavoodi@gmail.com} \and
Rajeev Raman\\
{University of Leicester}\\
\texttt{r.raman@leicester.ac.uk}\\
\and
{Srinivasa Rao Satti}\\
{Seoul National University}\\
\texttt{ssrao@cse.snu.ac.kr}
}
\maketitle
\begin{abstract}
We observe that a standard transformation between \emph{ordinal} trees
(arbitrary rooted trees with ordered children) and binary trees leads to interesting
succinct binary tree representations.  There are four symmetric versions of these
transformations.  Via these transformations we get four  
succinct representations of $n$-node binary trees  that
use $2n + n/(\log n)^{O(1)}$ bits and support (among other operations) navigation,
inorder numbering, one of pre- or post-order numbering, subtree size and 
lowest common ancestor (LCA) queries. The ability to support
inorder numbering is crucial for the well-known range-minimum query (RMQ) problem on
an array $A$ of $n$ ordered values.  
While this functionality, and more, is also 
supported in $O(1)$ time using $2n + o(n)$ bits
by Davoodi et al.'s (\emph{Phil. Trans. Royal Soc. A} \textbf{372} (2014)) extension of
a representation by Farzan and Munro (\emph{Algorithmica} \textbf{6} (2014)),
their \emph{redundancy}, or the $o(n)$ term, is much larger, and their approach may not be
suitable for practical implementations. 

One of these transformations is related to the Zaks' sequence (S.~Zaks, \emph{Theor.
Comput. Sci.} \textbf{10} (1980)) for encoding binary trees, and we thus provide the first
succinct binary tree representation based on Zaks' sequence. 
Another of these transformations is equivalent to Fischer and Heun's 
(\emph{SIAM J. Comput.} \textbf{40} (2011)) \minheap\ structure for this problem.
Yet another variant allows an encoding of 
the Cartesian tree of $A$ to be constructed from $A$ using only
$O(\sqrt{n} \log n)$ bits of working space.
\end{abstract}

%
%
%

\section{Introduction}
\label{sec:intro}
Binary trees are ubiquitous in computer science, and are integral to many 
applications that involve  indexing very large text collections.  
In such applications, the space usage of
the binary tree is an important consideration.  While a standard representation of a binary
tree node would use three pointers---to its left and right children, and to its parent---the
space usage of this representation, which is $\Theta(n\log n)$ bits to store an $n$-node tree, is unacceptable in
large-scale applications.  A simple counting
argument shows that the minimum space required to represent an $n$-node binary
tree is $2n - O(\log n)$ bits in the worst case.  
A number of
\emph{succinct} representations of static binary trees have been developed that support a wide range of operations
in $O(1)$ time\footnote{These results, and all others discussed here, assume the word RAM model with word size $\Theta(\log n)$ bits.}, 
using only $2n + o(n)$ bits \cite{jacobson89,fm-swat08}.

However, succinct binary tree representations have limitations.
A succinct binary tree representation can usually be divided into two parts:
a \emph{tree encoding}, which gives the structure of the tree, and
takes close to $2n$ bits, and an \emph{index} of $o(n)$ bits which is used
to perform operations on the given tree encoding.  It appears that the 
tree encoding constrains both the way nodes are numbered and the operations
that can be supported in $O(1)$ time using a $o(n)$-bit index.  For example,
the earliest succinct binary tree representation was due to Jacobson \cite{jacobson89},
but this only supported a level-order numbering, and while it supported basic
navigational operations such as moving to a parent, left child or right child in $O(1)$ time,
it did not support operations
such as \emph{lowest common ancestor (LCA)} and reporting the size of the subtree
rooted at a given node, in $O(1)$ time (indeed, it still remains unknown if there
is a $o(n)$-bit index to support these operations in Jacobson's encoding).

The importance of having a variety of node numberings and operations is 
illustrated by the \emph{range minimum
query (RMQ)} problem, defined as follows.
Given an array $A[1..n]$ of totally ordered values, RMQ 
problem is to preprocess $A$ into a data structure to answer the query $\rmqarg{i}{j}$: 
given two indexes  $1 \le i \le j \le n$, return the index of the minimum value in $A[i..j]$.
The aim is to minimize the query time, the space requirement of the data structure,
as well as the time and space requirements of the preprocessing. 
This problem finds variety of applications including range searching~\cite{Saxena2009}, text indexing~\cite{Abouelhoda2004,Sadakane2007}, text compression~\cite{Chen2008}, document retrieval~\cite{Muthukrishnan2002,Sadakane2007a,Valimaki2007}, flowgraphs~\cite{Georgiadis2004}, and position-restricted pattern matching~\cite{Iliopoulos2008}.  Since many of these applications deal with huge datasets, highly space-efficient solutions to the RMQ problem are of great interest.

A standard approach to solve the RMQ problem is via the~\emph{Cartesian tree}~\cite{Vuillemin1980}.
The Cartesian tree of $A$ is a binary tree obtained by labelling the root
with the value $i$ where $A[i]$ is the smallest
element in $A$.  The left subtree of the root is the Cartesian tree of
$A[1..i-1]$ and the right subtree of the root is the Cartesian tree of $A[i+1..n]$
(the Cartesian tree of an empty sub-array is the empty tree).   
As far as the RMQ problem is concerned, the key property
of the Cartesian tree is that the answer to the query $\rmqarg{i}{j}$ is the
label of the node that is the \emph{lowest common ancestor (LCA)} of
the nodes labelled $i$ and $j$ in the Cartesian tree.  Answering
RMQs this way does \emph{not} require access to $A$ at query time: this means
that $A$ can be (and often is) discarded after pre-processing. 
Since Farzan and
Munro \cite{fm-swat08} showed how to represent binary trees in $2n + o(n)$ bits
and support LCA queries in $O(1)$ time, it would appear that there is a fast and
highly space-efficient solution to the RMQ problem.  

Unfortunately, this is not true.  The difficulty is that the label 
of a node in the Cartesian tree (the index of the corresponding array element) is 
its rank in the \emph{inorder} traversal of the Cartesian tree. 
Until recently, none of the known binary tree representations \cite{jacobson89,fm-swat08,frs-icalp09}
was known to support inorder numbering.  Indeed, the first
$2n + o(n)$-bit and $O(1)$-time solution to the RMQ problem used an \emph{ordinal} tree,
or an arbitrary rooted, ordered tree, to answer RMQs \cite{fh-sjc11}.

Recently, Davoodi et al. \cite{NavarroDRS14} augmented the 
representation of \cite{fm-swat08} to support inorder numbering.  
However, Davoodi et al.'s result has some shortcomings.
The $o(n)$ additive term---the \emph{redundancy}---can be a significant overhead for practical 
values of $n$.
Using the results of \cite{DBLP:conf/focs/Patrascu08} (expanded upon in \cite{sn-soda10}), 
the redundancy of Jacobson's binary tree representation,
as well as Fischer and Heun's RMQ solution, can be reduced to $n/(\log n)^{O(1)}$:
this is not known to be true for the results of \cite{fm-swat08,NavarroDRS14}.  
Furthermore, there are several good practical implementations 
of ordinal trees \cite{acns09,ottaviano,sdsl},
but the approach of \cite{fm-swat08,NavarroDRS14} is complex and needs significant work
before its practical potential can even be evaluated.  
Finally, the results of \cite{fm-swat08,NavarroDRS14}
do not focus on the construction space or time, while this is considered in \cite{fh-sjc11}.

\subsubsection*{Our Results.}
We recall that there is a well-known transformation between binary trees and ordinal trees
(in fact, there are four symmetric versions of this transformation).
This allows us to represent a binary tree succinctly by transforming 
it into an ordinal tree, and then representing the ordinal tree succinctly.
We note a few interesting properties of the resulting binary tree representations:

\begin{itemize}
\item The resulting binary tree representations support inorder numbering in
addition to either postorder or preorder.
\item The resulting binary tree representations support a number of operations
including basic navigation, subtree size and LCA in $O(1)$ time; the latter implies
in particular that they are suitable for the RMQ problem.
\item The resulting binary tree representations use $2n + n/(\log n)^{O(1)}$ bits of
space, and so have low redundancy.
\item Since there are implementations of ordinal trees that are very fast in practice
\cite{acns09,ottaviano,sdsl}, we believe the resulting binary tree representations will
perform well in practice.
\item One of the binary tree representations, when applied to represent the Cartesian tree,
gives the same data structure as the \minheap\ of Fischer and Heun \cite{fh-sjc11}.
\item If one represents the ordinal tree using the BP encoding \cite{DBLP:conf/birthday/Raman013} then the
resulting binary tree encoding is \zak\ \cite{DBLP:journals/tcs/Zaks80}; we believe this to be the first succinct binary tree
representation based on \zak\ sequence.
\end{itemize}
Finally, we also show in Section~\ref{sec:construction-space} that using these representations,
we can make some improvements to the preprocessing phase of the RMQ problem.  Specifically,
we show that given an array $A$, a $2n+O(1)$-bit encoding of the tree structure of the
Cartesian tree of $A$ can be created in linear time using only $O(\sqrt{n} \log n)$ bits of 
working space.  This encoding can be augmented with additional data structures of 
size $o(n)$ bits, using only $o(n)$ bits of working space, thereby
``improving'' the result of~\cite{fh-sjc11} where $n+o(n)$ working 
space is used (the accounting of space is slightly different). 

\subsection{Preliminaries}
\label{sec:prelim}

\subsubsection*{Ordinal Tree Encodings.}
We now discuss two encodings of ordinal trees. The first encoding is the
natural \emph{balanced parenthesis (BP)} encoding~\cite{jacobson89,mr-sjc01}. 
The BP~sequence of an ordinal
tree is obtained by performing a depth-first traversal, and writing an opening parenthesis 
each time a node is visited, and a closing parenthesis immediately after all its descendants are visited. This gives a $2n$-bit encoding of an $n$-node ordinal
tree as a sequence of balanced parentheses.
The \emph{depth-first unary degree sequence (DFUDS)} \cite{bdmrrr-05}
is another encoding of an  $n$-node ordinal tree as a sequence of balanced parentheses.
This again visits the nodes in depth-first order (specifically, pre-order) 
and outputs the \emph{degree} $d$ of each node---defined here as the number of children
it has---in unary as $(^d )$. The resulting string is of length $2n -1$ bits
and is converted to a balanced parenthesis string by adding an open parenthesis 
to the start of the string.  See Figure~\ref{fig:tree-encodings} for an example.

\begin{figure}[t]
  \centering
  \includegraphics[scale=0.4]{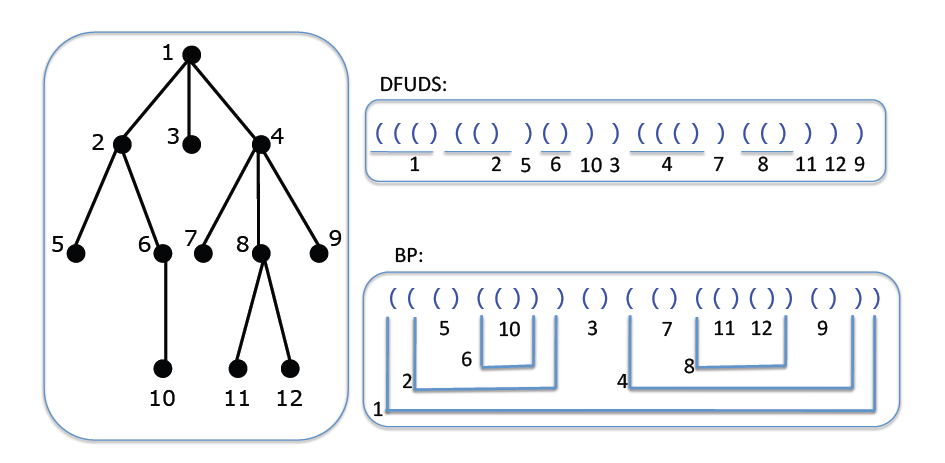}
\caption{BP and DFUDS encodings of ordinal trees.}
\label{fig:tree-encodings}
\end{figure}

\subsubsection*{Succinct Ordinal Tree Representations.}
Table~\ref{tab:ops} gives a subset of operations that are supported by 
various ordinal tree representations (see \cite{DBLP:conf/birthday/Raman013} for other operations).
\begin{table}
\caption{A list of operations on ordinal trees.  In all cases below, if the value returned by the operation is not defined (e.g. performing $\parent()$ on the root of the tree) an appropriate ``null'' value is returned.}
\label{tab:ops}
\centering
\begin{tabular}{|l|l|}
\hline
\multicolumn{1}{|c|}{Operation} & \multicolumn{1}{|c|}{Return Value}\\
\hline 
$\parent(x)$  & the parent of node $x$\\
$\child(x,i)$ & the $i$-th child of node $x$, for $i\geq1$\\
$\nextsiblingarg{x}$ & the next sibling of $x$\\ 
$\prevsiblingarg{x}$ & the previous sibling of $x$\\ 
$\depth(x)$ &  the depth of node $x$\\
$\select_{o}(j)$ & the $j$-th node in $o$-order, for $o \in \{ \preorder, \postorder,$ \\
                 & \hspace*{1in}$\preorderRight, \postorderRight\}$\\
$\lmleaf(x)$ & the leftmost leaf of the subtree rooted at node $x$\\
$\rmleaf(x)$ & the rightmost leaf of the subtree rooted at node $x$\\ 
$\subtreesizearg{x}$ & the size of the subtree rooted at node $x$ (excluding $x$ itself)\\
$\lca(x, y)$ & the lowest common ancestor of the nodes $x$ and $y$\\
$\levelancarg{x}{i}$ & the ancestor of node $x$ at depth $d-i$, 
where $d = \depth(x)$, for $i\geq0$\\
\hline
\end{tabular}
\end{table}
%
%
\begin{theorem}[\cite{sn-soda10}] 
\label{thm:ordinal-tree-rep}
There is a succinct ordinal tree representation that supports all operations in Table~\ref{tab:ops} in $O(1)$ time and uses $2n + n/(\log n)^{O(1)}$ bits of space, where $n$ denotes the number of nodes in the represented ordinal tree.
%
\end{theorem}

\section{Succinct Binary Trees Via Ordinal Trees}
\label{sec:transform}
We describe succinct binary tree representations which are based on ordinal tree representations. In other words, we present a method that transforms a binary tree to an ordinal tree (indeed we study  several and similar transformations) with properties used to simulate various operations on the original binary tree. We show that using known succinct ordinal tree representations from the literature, we can build succinct binary tree representations that have not yet been discovered. We also introduce a relationship between two standard ordinal tree encodings (BP and DFUDS) and study how this relationship can be utilized in our binary tree representations.

\subsection{Transformations}
We define four transformations \trans{1}, \trans{2}, \trans{3}, and \trans{4} 
that transform a binary tree into an ordinal tree.
We describe the $i$-th transformation by stating how we generate the output ordinal tree \transTree{i} given an input binary tree \bintree. Let $n$ be the number of nodes in \bintree. The number of nodes in each of \transTree{1}, \transTree{2}, \transTree{3}, and \transTree{4} is $n+1$, where each node corresponds to a node in \bintree\ except the root (think of the root as a dummy node). In the following, we show the correspondence between the nodes in \bintree\ and the nodes in \transTree{i} by using the notation \transOrdering{i}{$u$}{$v$}, which means the node $u$ in \bintree\ corresponds to the node $v$ in \transTree{i}. Given a node $u$ in \bintree, and its corresponding node $v$ in \transTree{i}, we show which nodes in \transTree{i} correspond to the left and right children of $u$:
\\

\setlength{\tabcolsep}{0.6cm}
\begin{tabular}[h]{l l}
\transOrdering{1}{left-child$(u)$}{first-child$(v)$}
&
\transOrdering{1}{right-child$(u)$}{next-sibling$(v)$}
\\

\transOrdering{2}{left-child$(u)$}{last-child$(v)$}
&
\transOrdering{2}{right-child$(u)$}{previous-sibling$(v)$}
\\

\transOrdering{3}{left-child$(u)$}{next-sibling$(v)$}
&
\transOrdering{3}{right-child$(u)$}{first-child$(v)$}
\\

\transOrdering{4}{left-child$(u)$}{previous-sibling$(v)$}
&
\transOrdering{4}{right-child$(u)$}{last-child$(v)$}
\\
\end{tabular}
\\

The example in Figure \ref{fig:transformations} shows a binary tree which is transformed to an ordinal tree by each of the four transformations. Notice that \transTree{2} and \transTree{4} are the mirror images of \transTree{1} and \transTree{3}, respectively.

\begin{figure}[t]
  \centering
  \includegraphics[scale=0.43]{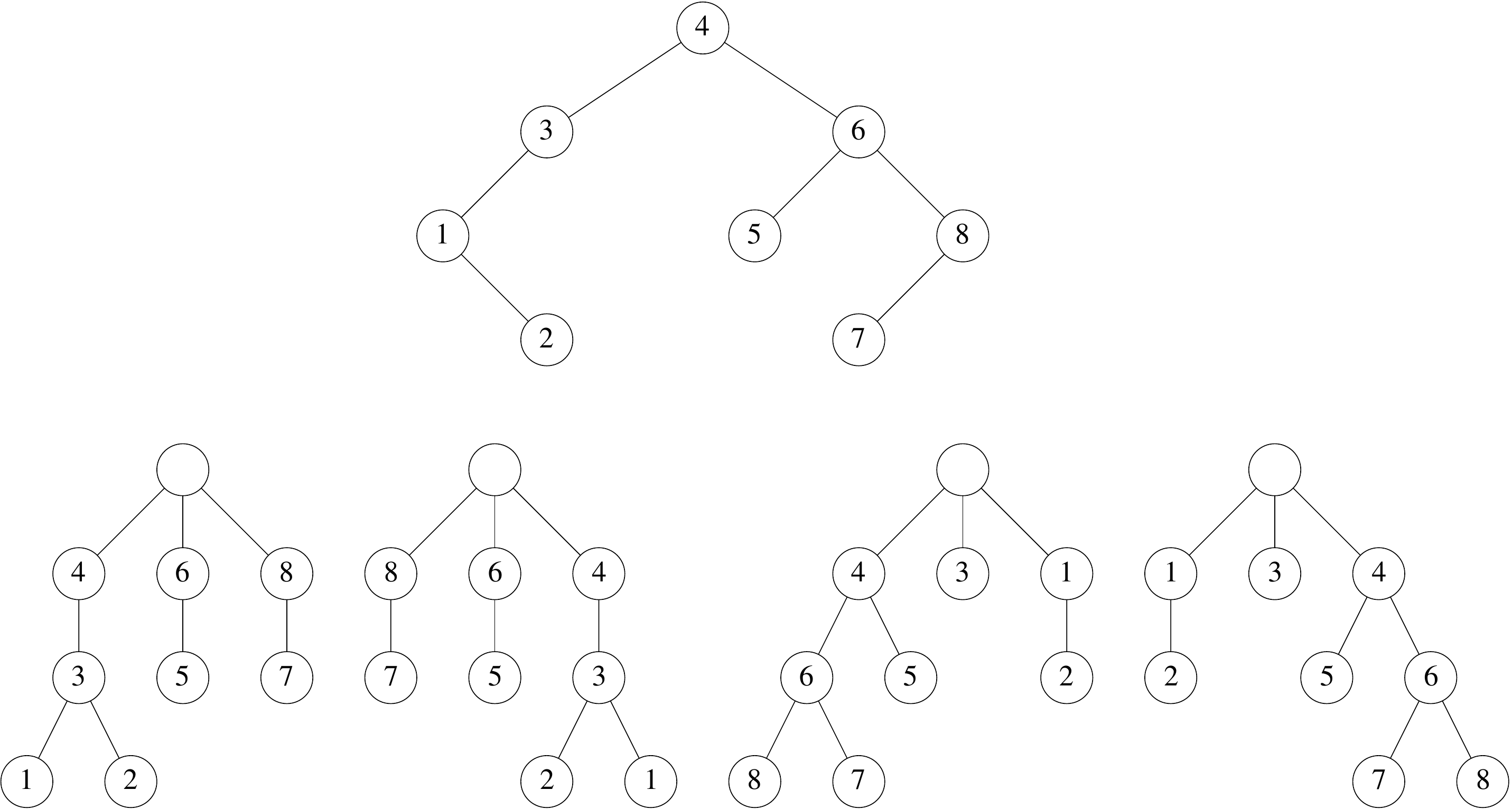}
  \caption{Top: a binary tree, nodes numbered in inorder. Bottom from left to right: the ordinal trees \transTree{1}, \transTree{2}, \transTree{3}, \transTree{4}. A node with number $k$ in the binary tree corresponds to the node with the same number $k$ in the ordinal trees. Notice that the roots of the ordinal trees do not correspond to any node in the binary tree.}
  \label{fig:transformations}
\end{figure}

\subsection{Effect of Transformations on Orderings}
It is clear that the order in which the nodes appear in \bintree\ is different from the node ordering in each of \transTree{1}, \transTree{2}, \transTree{3}, and \transTree{4}; however there is a surprising relation between these orderings. Consider the three standard orderings \inorder, \preorder, and \postorder\ on binary trees and ordinal trees\footnote{\inorder\ is only defined on binary trees}. While \preorder\ and \postorder\ arrange the nodes from left to right by default, we define their symmetries which arrange the nodes from right to left: \preorderRight\ and \postorderRight\ (i.e., visiting the children of each node from right to left in traversals). Each of the four transformations maps two of the binary tree orderings to two of the ordinal tree orderings. For example, if the first transformation maps \inorder\ to \postorder, that means a node with \inorder\ rank $k$ in \bintree\ corresponds to a node with \postorder\ rank $k$ in \transTree{1}. In the following, we show all of these mappings ($ -1 $ means that the mapping is off by 1, due to the presence of the dummy root): 

\setlength{\tabcolsep}{0.5cm}
\begin{tabular}{l l l}
\trans{1}: 
& 
$\inorder\  \rightarrow   \postorder$
&
$\preorder\  \rightarrow   \preorder\  -1 $
\\

\trans{2}: 
& 
$\inorder\  \rightarrow   \postorderRight$
&
$\preorder\  \rightarrow   \preorderRight\  -1 $
\\

\trans{3}: 
& 
$\inorder\  \rightarrow   \preorderRight -1$
&
$\postorder\  \rightarrow   \postorderRight $
\\

\trans{4}: 
& 
$\inorder\  \rightarrow   \preorder -1$
&
$\postorder  \rightarrow   \postorder $
\\
\end{tabular}

\subsection{Effect of Transformations on Ordinal Tree Encodings}
The four transformation methods define four ordinal trees \transTree{1}, \transTree{2}, \transTree{3}, \transTree{4} (from a given binary tree) which may look completely different, however they have relationships. In the following, we show a pairwise connection between \transTree{1} and \transTree{3} and between \transTree{2} and \transTree{4}.

\paragraph{From paths to siblings and vice versa.}
Recall the following transformation functions on the binary tree nodes:

\setlength{\tabcolsep}{1cm}
\begin{tabular}[h]{l l}
\transOrdering{1}{left-child$(u)$}{first-child$(v)$}
&
\transOrdering{1}{right-child$(u)$}{next-sibling$(v)$}
\\

\transOrdering{3}{left-child$(u)$}{next-sibling$(v)$}
&
\transOrdering{3}{right-child$(u)$}{first-child$(v)$}
\end{tabular}
\\
If we combine the above functions, we can conclude the following:
$$
\text{first-child$ (v) $ in \trans{1}} = \text{next-sibling$ (v) $ in \trans{3}}
$$
$$
\text{next-sibling$ (v) $ in \trans{1}} = \text{first-child$ (v) $ in \trans{3}}
$$
The above relation yields a connection between paths in \transTree{1} and the siblings in \transTree{3}. Consider any downward path $(u_1, u_2, \ldots, u_k)$ in \transTree{1}, where $u_1$ is the $i$-th child of a node for any $i>1$ and all $u_2, \ldots, u_k$ are the first child of their parents. All the nodes $u_1, u_2, \ldots, u_k$ become siblings in \transTree{3} and they are all the children of a node which corresponds to the previous-sibling of $u_1$ in \transTree{1}. Also, if $u_1$ is the root of \transTree{1} then $u_1$ corresponds to the root in \transTree{3} and all $u_2, \ldots, u_k$ become the children of the root of \transTree{3}. The paths in \transTree{3} are related to siblings in \transTree{1} in a similar way. Also, a similar connection exists between \transTree{2} and \transTree{4} with the only difference that first child, next sibling, and previous sibling should be turned into last child, next sibling, and previous sibling in order.

\paragraph{From BP to \dfuds\ and vice versa.}
Section~\ref{sec:prelim} introduced two standard ordinal tree encodings: BP and \dfuds. 
The relation between paths and siblings in the different types of transforms 
suggests a relation between the BP and \dfuds\ sequences of these transformations,
as paths of the kind considered above lead to consecutive sequences of parentheses, 
which in \dfuds\ can be viewed as the encoding of a group of siblings.
We now formalize this intuition. Let \dfudsRight\ be the \dfuds\ sequence where the children of every node are visited from right to left.

\begin{figure}
  \centering
  \includegraphics[scale=0.5]{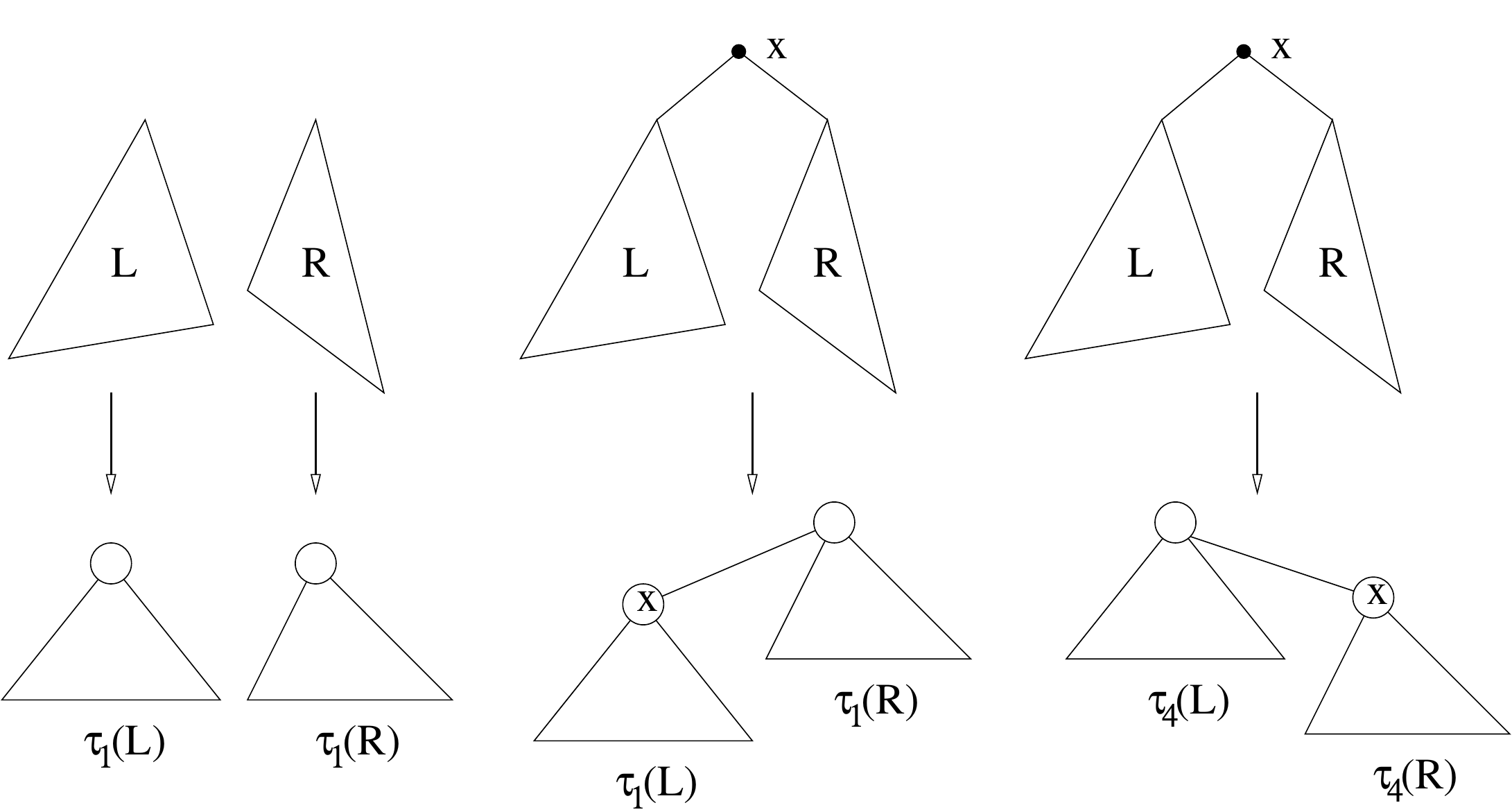}
  \caption{Illustrating how the result of transformations \trans{1} and \trans{4} on a binary tree with root $x$, left subtree $L$ and right subtree $R$ can be obtained from the transformations \trans{1} and \trans{4} on $L$ and $R$.}
  \label{fig:join}
\end{figure}

\begin{theorem}
\label{thm:equivalence}
For any binary tree $T_b$, let
\transTree{i} denote the result of \trans{i} applied to $T_b$. Then:
\begin{enumerate}
\item
The BP sequences of \transTree{1} and \transTree{3}
are the reverses of the BP sequences of \transTree{2} and \transTree{4}, respectively.

\item The \dfuds\ sequences of \transTree{1} and \transTree{3} are equal to the
\dfudsRight\ sequences of \transTree{2} and \transTree{4}, respectively.

\item
The BP sequence of \transTree{1} is equal to the \dfuds\ sequence of \transTree{4}.

\item
The BP sequence of \transTree{2} is equal to the \dfuds\ sequence of \transTree{3}.
\end{enumerate}
\end{theorem}
\begin{proof} (1) and (2) follow directly from the definitions of the transformations.

We prove (3) by induction ((4) is analogous). For
any ordinal tree $T$ denote by BP($T$) and \dfuds($T$) the BP and
\dfuds\ encodings of $T$.  For the base case, if $T_b$ consists of a singleton node, then 
BP(\transTree{1}) = \dfuds(\transTree{4}) = (()).  
Before going to the inductive case, observe that for any binary tree
$T$, BP(\trans{i}($T$)) is a string of the form $(Y)$, where the parenthesis 
pair enclosing $Y$ represents the dummy root of 
\trans{i}($T$), and $Y$ is the BP representation of the
ordered forest obtained by deleting the dummy root.  
On the other hand, \dfuds(\trans{i}(T)), which can also
be viewed as a string of the form $(Y)$, is interpreted differently:
the first open parenthesis is a dummy, while $Y)$ is the \emph{core} \dfuds\ representation
of the entire tree \trans{i}($T$), including the dummy root. 

Now let $T_b$ be a binary tree with root $x$, left subtree $L$ and right subtree $R$. 
By induction, BP(\trans{1}($L$)) = \dfuds(\transTree{4}($L$)) = $(Y)$ and
BP(\transTree{1}($R$)) = \dfuds(\transTree{4}($R$)) = $(Z)$.  
Note that \transTree{1} is obtained as follows.  Replace the dummy root of
\trans{1}($L$) by $x$, and insert $x$ as the first child
of the dummy root of \trans{1}($R$) (see Figure~\ref{fig:join}). Thus, BP(\transTree{1}) = $((Y)Z)$.

Next, \transTree{4} is obtained as follows: replace the dummy root of
\trans{4}($R$) by $x$, and insert $x$ as the last child of the dummy root
of \trans{4}($L$) (see Figure~\ref{fig:join}).  
The core \dfuds\ representation of \transTree{4}
(i.e. excluding the dummy open parenthesis) is obtained as follows: add an 
open parenthesis to the front of the core \dfuds\ representation of \trans{4}($L$),
to indicate that the root of \trans{4}($L$), which is also the root of
\transTree{4}, has one additional child.  This gives the parenthesis sequence
$(Y)$.  To this we append the
core \dfuds\ representation of \trans{4}($R$), giving $(Y)Z)$.  
Now we add the dummy open parenthesis, giving
$((Y)Z)$, which is the same as BP(\transTree{1}).
\end{proof}

\subsection{Succinct Binary Tree Representations}
We consider the problem of encoding a binary tree in a succinct data structure that supports the operations: \leftchild, \rightchild, \parent, \subtreesize, \lca, \inorder, and one of \preorder\ or \postorder\ (we give two data structures supporting \preorder\ and two other ones supporting \postorder). We design such a data structure by transforming the input binary tree into an ordinal tree using any of the transformations \trans{1}, \trans{2}, \trans{3}, or \trans{4}, and then encoding the ordinal tree by a succinct data structure that supports appropriate operations. For each of the four ordinal trees that can be obtained by the transformations, we give a set of algorithms, each performs one of the binary tree operations. In other words, we show how each of the binary tree operations can be performed by ordinal tree operations. Here, we show this for transformation \trans{1}, and later in the section we present the pseudocode for all the transformations. In the following, $u$ denotes any given node in a binary tree \bintree, and \transarg{1}{u} denotes the node corresponding to $u$ in \transTree{1}:

\paragraph{\leftchild.}
The left child of $u$ is the first child of \transarg{1}{u}, which can be determined by the operation \childarg{1}{\transarg{1}{u}} on \transTree{1}. 

\paragraph{\rightchild.}
The right child of $u$ is the next sibling of \transarg{1}{u}, which can be determined by the operation \nextsiblingarg{\transarg{1}{u}} on \transTree{1}.

\paragraph{\parent.}
To determine the parent of $u$, there are two cases: 
(1) if \transarg{1}{u} is the first child of its parent, then the answer is the parent of \transarg{1}{u} to be determined by \parentarg{\transarg{1}{u}} on \transTree{1}; 
(2) if \transarg{1}{u} is not the first child of its parent, then the answer is the previous sibling of \transarg{1}{u} to be determined by the operation \prevsiblingarg{\transarg{1}{u}} on \transTree{1}.

\paragraph{subtree-size.} 
The subtree size of $u$ is equal to the subtree size of \transarg{1}{u} plus the sum of the subtree sizes of all the siblings to the right of \transarg{1}{u}.
Let $\ell$ be the rightmost leaf in the subtree of the parent of \transarg{1}{u}. To obtain the above sum, we subtract the preorder number of \transarg{1}{u} from the preorder number of~$\ell$, which can be performed using the operations \rmleafarg{\parentarg{\transarg{1}{u}}} and \preorder.

\paragraph{LCA}
Let $w$ be the LCA of the two nodes $u$ and $v$ in \bintree. We want to find $w$ given $u$ and $v$, assuming w.l.o.g. that \preorderarg{u} is smaller than \preorderarg{v}. Let lca be the LCA of \transarg{1}{u} and \transarg{1}{v} in \transTree{1}. Observe that lca is a child of \transarg{1}{w} and an ancestor of \transarg{1}{u}. Thus, we only need to find the ancestor of \transarg{1}{u} at level $i$, where $i-1$ is the depth of lca. To compute this, we utilize the operations \lca, \depth, and \levelanc\ on \transTree{1}.

See Tables \ref{tbl:operations-t1} through 
\ref{tbl:operations-t4} for the pseudocode of all the binary tree operations on all four transformations.
\begin{table}[f]
\caption{Binary tree operations performed using ordinal tree operations for transformation \trans{1}. The ordinal tree operations \preorder, \preorderRight\ \postorder, and \postorderRight\ are the same as $\select_{o}$ for different values of $ o $ (refer to Table \ref{tab:ops}).}
\label{tbl:operations-t1}
\begin{tabular}{ll}

\begin{minipage}{5cm}
\begin{alltt}
\inorderarg{u}:
    return \postorderarg{\transarg{1}{u}}
\end{alltt}
\end{minipage}
&
\begin{minipage}{5cm}
\begin{alltt}
\preorderarg{u}:
    return \preorderarg{\transarg{1}{u}} - 1 
\end{alltt}
\end{minipage}

\\
\\

\begin{minipage}{5cm}
\begin{alltt}
\leftchildarg{u}:
    return \childarg{1}{\transarg{1}{u}}
\end{alltt}
\end{minipage}
&
\begin{minipage}{5cm}
\begin{alltt}
\rightchildarg{u}:
    return \nextsiblingarg{\transarg{1}{u}}
\end{alltt}
\end{minipage}

\\
\\

\begin{minipage}{5cm}
\begin{alltt}
\parentarg{u}:
    \ensuremath{v} = \prevsiblingarg{\transarg{1}{u}}
    if \ensuremath{v} \ensuremath{\neq} NULL
        return \ensuremath{v}
    else
        return \parentarg{\transarg{1}{u}}
\end{alltt}

\begin{alltt}
\subtreesizearg{u}:
    \ensuremath{\ell} = \rmleafarg{\parentarg{\transarg{1}{u}}}
    return \preorderarg{\ensuremath{\ell}} - \preorderarg{\transarg{1}{u}}
\end{alltt}
\end{minipage}
&
\begin{minipage}{8cm}
\begin{alltt}
//assume \preorderarg{u} < \preorderarg{v}
\lcaarg{u}{v}:
    lca = \lcaarg{\transarg{1}{u}}{\transarg{1}{v}}
    if lca == \transarg{1}{u}
        return lca
    if lca == \parentarg{\transarg{1}{u}}
        return \transarg{1}{u}
    \ensuremath{i} = \deptharg{\text{lca}}+1
    return \levelancarg{\transarg{1}{u}}{i}
\end{alltt}
\end{minipage}

\end{tabular}
\end{table}

\begin{table}[f]
\caption{Binary tree operations performed using ordinal tree operations for transformation \trans{2}.  The ordinal tree operations \preorder, \preorderRight\ \postorder, and \postorderRight\ are the same as $\select_{o}$ for different values of $ o $ (refer to Table \ref{tab:ops}).}
\label{tbl:operations-t2}
\begin{tabular}{ll}

\begin{minipage}{5cm}
\begin{alltt}
\inorderarg{u}:
    return \postorderRightarg{\transarg{2}{u}}
\end{alltt}
\end{minipage}
&
\begin{minipage}{5cm}
\begin{alltt}
\preorderarg{u}:
    return \preorderRightarg{\transarg{2}{u}} - 1 
\end{alltt}
\end{minipage}

\\
\\

\begin{minipage}{5cm}
\begin{alltt}
\leftchildarg{u}:
    \ensuremath{d} = \degreearg{\transarg{2}{u}}
    return \childarg{d}{\transarg{2}{u}}
\end{alltt}
\end{minipage}
&
\begin{minipage}{5cm}
\begin{alltt}
\rightchildarg{u}:
    return \prevsiblingarg{\transarg{2}{u}}
\end{alltt}
\end{minipage}

\\
\\

\begin{minipage}{5cm}
\begin{alltt}
\parentarg{u}:
    \ensuremath{v} = \nextsiblingarg{\transarg{2}{u}}
    if \ensuremath{v} \ensuremath{\neq} NULL
        return \ensuremath{v}
    else
        return \parentarg{\transarg{2}{u}}
\end{alltt}

\begin{alltt}
\subtreesizearg{u}:
    \ensuremath{\ell} = \lmleafarg{\parentarg{\transarg{2}{u}}}
    return \preorderRightarg{\ensuremath{\ell}} -
                \preorderRightarg{\transarg{2}{u}}
\end{alltt}
\end{minipage}
&
\begin{minipage}{8cm}
\begin{alltt}
//assume \preorderarg{u} > \preorderarg{v}
\lcaarg{u}{v}:
    lca = \lcaarg{\transarg{2}{u}}{\transarg{2}{v}}
    if lca == \transarg{2}{u}
        return lca
    if lca == \parentarg{\transarg{2}{u}}
        return \transarg{2}{u}
    \ensuremath{i} = \deptharg{\text{lca}}+1
    return \levelancarg{\transarg{2}{u}}{i}
\end{alltt}
\end{minipage}

\end{tabular}
\end{table}


\begin{table}[f]
\caption{Binary tree operations performed using ordinal tree operations for transformation \trans{3}.  The ordinal tree operations \preorder, \preorderRight\ \postorder, and \postorderRight\ are the same as $\select_{o}$ for different values of $ o $ (refer to Table \ref{tab:ops}).}
\label{tbl:operations-t3}
\begin{tabular}{ll}

\begin{minipage}{5cm}
\begin{alltt}
\inorderarg{u}:
    return \preorderRightarg{\transarg{3}{u}} - 1 
\end{alltt}
\end{minipage}
&
\begin{minipage}{5cm}
\begin{alltt}
\postorderarg{u}:
    return \postorderarg{\transarg{3}{u}}
\end{alltt}
\end{minipage}

\\
\\

\begin{minipage}{5cm}
\begin{alltt}
\leftchildarg{u}:
    return \nextsiblingarg{\transarg{3}{u}}
\end{alltt}
\end{minipage}
&
\begin{minipage}{5cm}
\begin{alltt}
\rightchildarg{u}:
    return \childarg{1}{\transarg{3}{u}}
\end{alltt}
\end{minipage}

\\
\\

\begin{minipage}{5cm}
\begin{alltt}
\parentarg{u}:
    \ensuremath{v} = \prevsiblingarg{\transarg{3}{u}}
    if \ensuremath{v} \ensuremath{\neq} NULL
        return \ensuremath{v}
    else
        return \parentarg{\transarg{3}{u}}
\end{alltt}

\begin{alltt}
\subtreesizearg{u}:
    \ensuremath{\ell} = \rmleafarg{\parentarg{\transarg{3}{u}}}
    return \postorderRightarg{\ensuremath{u}} -
               \postorderRightarg{\transarg{3}{\ell}}
\end{alltt}
\end{minipage}
&
\begin{minipage}{8cm}
\begin{alltt}
//assume \preorderarg{u} < \preorderarg{v}
\lcaarg{u}{v}:
    lca = \lcaarg{\transarg{3}{u}}{\transarg{3}{v}}
    if lca == \transarg{3}{u}
        return lca
    if lca == \parentarg{\transarg{3}{u}}
        return \transarg{3}{u}
    \ensuremath{i} = \deptharg{\text{lca}}+1
    return \levelancarg{\transarg{3}{u}}{i}
\end{alltt}
\end{minipage}

\end{tabular}
\end{table}

\begin{table}[f]
\caption{Binary tree operations performed using ordinal tree operations for transformation \trans{4}.  The ordinal tree operations \preorder, \preorderRight\ \postorder, and \postorderRight\ are the same as $\select_{o}$ for different values of $ o $ (refer to Table \ref{tab:ops}).}
\label{tbl:operations-t4}
\begin{tabular}{ll}

\begin{minipage}{5cm}
\begin{alltt}
\inorderarg{u}:
    return \preorderarg{\transarg{4}{u}} - 1 
\end{alltt}
\end{minipage}
&
\begin{minipage}{5cm}
\begin{alltt}
\postorderarg{u}:
    return \postorderarg{\transarg{4}{u}}
\end{alltt}
\end{minipage}

\\
\\

\begin{minipage}{5cm}
\begin{alltt}
\leftchildarg{u}:
    return \prevsiblingarg{\transarg{4}{u}}
\end{alltt}
\end{minipage}
&
\begin{minipage}{5cm}
\begin{alltt}
\rightchildarg{u}:
    \ensuremath{d} = \degreearg{\transarg{4}{u}}
    return \childarg{d}{\transarg{4}{u}}\end{alltt}
\end{minipage}

\\
\\

\begin{minipage}{5cm}
\begin{alltt}
\parentarg{u}:
    \ensuremath{v} = \nextsiblingarg{\transarg{4}{u}}
    if \ensuremath{v} \ensuremath{\neq} NULL
        return \ensuremath{v}
    else
        return \parentarg{\transarg{4}{u}}
\end{alltt}

\begin{alltt}
\subtreesizearg{u}:
    \ensuremath{\ell} = \lmleafarg{\parentarg{\transarg{4}{u}}}
    return \postorderarg{\ensuremath{u}} -
                \postorderarg{\transarg{4}{\ell}}
\end{alltt}
\end{minipage}
&
\begin{minipage}{8cm}
\begin{alltt}
//assume \preorderarg{u} > \preorderarg{v}
\lcaarg{u}{v}:
    lca = \lcaarg{\transarg{4}{u}}{\transarg{4}{v}}
    if lca == \transarg{4}{u}
        return lca
    if lca == \parentarg{\transarg{4}{u}}
        return \transarg{4}{u}
    \ensuremath{i} = \deptharg{\text{lca}}+1
    return \levelancarg{\transarg{4}{u}}{i}
\end{alltt}
\end{minipage}

\end{tabular}
\end{table}

\begin{theorem}
\label{thm:transform}
There is a succinct binary tree representation of size $2n + n/(\log n)^{O(1)}$ bits that supports all the operations \leftchild, \rightchild, \parent, \subtreesize, \lca, \inorder, and one of \preorder\ or \postorder\ in $O(1)$ time, where $n$ denotes the number of nodes in the represented binary tree.
\end{theorem}
\begin{proof}
We can use any of the transformations \trans{1}, \trans{2}, \trans{3}, or \trans{4} and use any ordinal tree representation that uses $2n + n/(\log n)^{O(1)}$ bits and supports the operations required by our algorithms (Theorem~\ref{thm:ordinal-tree-rep}).
\end{proof}

\subsection{BP for Binary Trees (\zak)}
\label{sec:binary-tree-bp}
Let \bintree\ be a binary tree which is transformed to the ordinal tree \transTree{1} by the first transformation \trans{1}. We show that the BP sequence of \transTree{1} is a special sequence for \bintree. This sequence is called \zak, and our transformation methods show that \zak\ can be indeed used to encode binary trees into a succinct data structure that supports various operations. In the following, we give a definition for \zak\ of \bintree\ and we show that it is indeed equal to the BP sequence of~\transTree{1}.

\paragraph{\zak.}
Initially, extend the binary tree \bintree\ by adding external nodes wherever there is a missing child. Now, label the internal nodes with an open parenthesis, and the external nodes with a closing parenthesis. \zak\ of \bintree\ is derived by traversing \bintree\ in \preorder\ and printing the labels of the visited nodes. If \bintree\ has $n$ nodes, \zak\ is of length $2n +1$. If we insert an open parenthesis at the beginning of \zak, then it becomes a balanced-parentheses sequence. Notice that \zak\ is different from Jacobson's~\cite{jacobson89} approach where \bintree\ is traversed in level-order. See Figure \ref{fig:binary-tree-bp} for an example.

Each pair of matching open and close parentheses in \zak, except the extra open parenthesis at the beginning and its matching close parenthesis, is (conceptually) associated with a node in \bintree. Observe that the open parentheses in the sequence from left to right correspond to the nodes in preorder, and the closing parentheses from left to right correspond to the nodes in inorder.

\begin{figure}[t]
  \centering
  \includegraphics[scale=0.5]{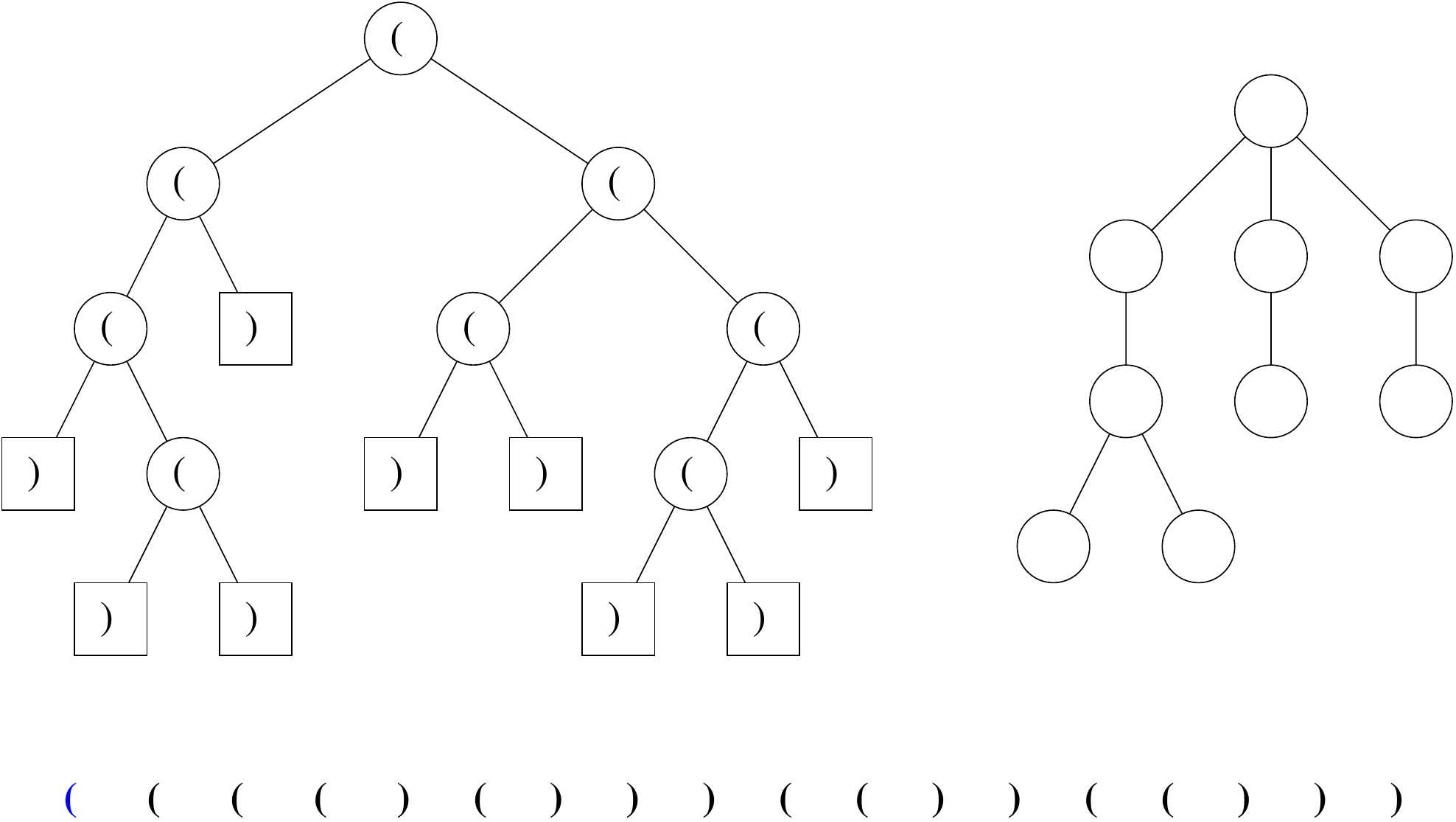}
  \caption{Illustrating the connection between \zak\ of \bintree\ and the BP sequence of \transTree{1} (Left: A binary tree \bintree; Right: The ordinal tree \transTree{1} derived from applying transformation \trans{1} to \bintree).} 
  \label{fig:binary-tree-bp}
\end{figure}

\paragraph{\zak\ and transformation \trans{1}.}
\zak\ of \bintree\ including the extra open parenthesis at the beginning is the same as the BP sequence of \transTree{1}. This is stated in the following lemma:

\begin{theorem}
An open parenthesis appended by \zak\ of a binary tree is the same as the BP sequence of the ordinal tree obtained by applying the first transformation \trans{1} to the binary tree.
\end{theorem}
\begin{proof}
Let \bintree\ and \transTree{1} be the binary tree and ordinal tree respectively, and let $S$ be the sequence of an open parenthesis appended by \zak\ of \bintree. We prove the lemma by induction on the number of nodes in \bintree.

For the base case,
If \bintree\ has only one node, then \zak\ of \bintree\ will be $ ()) $ and the BP of \transTree{1} will be $ (()) $.  Take as the induction hypothesis that the 
lemma is correct for \bintree\ with less than $ n $ nodes.

In the inductive step, let \bintree\ have $ n $ nodes. Let $u_1, u_2, \ldots, u_k$ be the right most path in \bintree, where $ 1\le k \le n $. Let $l_i$ be the left subtree of $ u_i $ for $ 1\le i \le k $. Thus, $ S = ( (z_1 (z_2 \cdots (z_k ) $, where $ z_i $ is \zak\ of $ l_i $. Notice that the size of each $ l_i $ is smaller than $ n $. Therefore by the hypothesis, $(z_i $ is the BP sequence $b_i$ of the ordinal tree \transarg{1}{l_i}. The ordinal tree \transTree{1} has a dummy root which has $ k $ subtrees which are identical to \transarg{1}{l_i} for $1\le i \le k$. Therefore, the BP sequence of \transTree{1} is $(b_1 b_2 \cdots b_k) $, where $ b_i$ is the BP sequence of the subtree rooted at \transarg{1}{u_i}. Thus, the BP sequence of \transTree{1} is the same as $ S $.
\end{proof}

\section{Cartesian Tree Construction in $o(n)$ Working Space}
\label{sec:construction-space}

Given an array $A$, we now show how to construct one of the succinct tree representations 
implied in Theorem~\ref{thm:transform} in small working space.
The model is that the array $A$ is present in read-only memory.  
There is a working memory, which is read-write, and the succinct 
Cartesian tree representation is written to a readable, but 
write-once, output memory.  All memory is assumed to be randomly accessible.

%
%
A straightforward way to construct a succinct representation of a Cartesian tree is to construct the standard pointer-based representation of the Cartesian tree from the given array in linear time~\cite{gbt-stoc84}, and then construct the succinct representation using the pointer-based representation. The drawback of this approach is that the space used during the construction is $O(n \log n)$ bits, although the final structure uses only $O(n)$ bits.  Fischer and Heun~\cite{fh-sjc11} show that the construction space can be reduced to $n + o(n)$ bits. In this section, we show how to improve the construction space to $o(n)$ bits.  In fact, we first show that a parenthesis sequence
corresponding to the Cartesian tree can be output in linear time using only $O(\sqrt{n} \log n)$ bits
of working space; subsequently, using methods from \cite{grrr-tcs06} the auxiliary data structures can also
be constructed in $o(n)$ working space.


\begin{theorem}
Given an array $A$ of $n$ values, let $T_c$ be the Cartesian tree of $A$.
We can output BP(\trans{4}($T_c$)) in $O(n)$ time 
using $O(\sqrt{n} \log n)$ bits of working space.
\end{theorem}
\begin{proof}
The algorithm reads $A$ from left to right and outputs a
parenthesis sequence as follows: when processing $A[i]$, for
$i = 1, \ldots, n$ we compare 
$A[i]$ with all the suffix minima of $A[1..i-1]$---if $A[i]$ is smaller
than $j \ge 0$ suffix minima, then we output the string $(\, )^j$.  Finally
we add the string $(\,)^{j+1}$ to the end, where $j$ is the number of suffix 
minima of $A[1..n]$.  

Given any ordinal tree $T$, define the \emph{post-order degree sequence}
parenthesis string obtained from $T$, denoted by PODS($T$), as follows.  
Traverse $T$ in post-order, and at each node 
that has $d$ children, output the bit-string $(\, )^d$.  At the end, output
an additional $)$. Let $T_c$ be the Cartesian tree of $A$ and for $i = 1, \ldots, 4$, 
let \transTree{i} = \trans{i}($T_c$) as before.  Observe that in 
\transTree{1}, for $i=1, \ldots, n$, 
the children of the $i$-th node in post-order, which is the
$i$-th node in $T_c$ in in-order and hence corresponds to $A[i]$,
are precisely those suffix minima that $A[i]$ was smaller than when
$A[i]$ was processed.  Furthermore, the children of the (dummy) root of
\transTree{1} are the suffix maxima of $A[1..n]$. Thus,
the string output by the above pre-processing of $A$ is 
PODS(\transTree{1}).

We now show that PODS(\transTree{1}) = BP(\transTree{4}).  The proof
is along the lines of Theorem~\ref{thm:equivalence}.  If $L$ and $R$ denote
the left subtree of the root of $T_c$, assume inductively that
BP(\trans{4}($L$)) = PODS(\trans{1}($L$)) = $(Y)$ and 
BP(\trans{4}($R$)) = PODS(\trans{1}($L$)) = $(Z)$.
Using the reasoning in Theorem~\ref{thm:equivalence} (summarized
in Figure~\ref{fig:join}) we see immediately 
that BP(\transTree{4}) = $(Y(Z))$, and by a reasoning similar to
the one used for \dfuds\ in Theorem~\ref{thm:equivalence}, it is
easy to see that PODS(\transTree{1}) = $(Y(Z))$ as well, as required.

While the straightforward approach would be to maintain a linked list of the locations of the current suffix minima, this list could contain $\Theta(n)$ locations and could take $\Theta(n \log n)$ bits.  
Our approach is to use the output string itself to encode the positions of the
suffix minima.  One can observe that if the output string is created by
the above process, it will be of the form $( b_1 ( ... ( b_k $ where 
each $b_i$ is a (possibly empty) maximal balanced parenthesis string -- the remaining
parentheses are called \emph{unmatched}.  It is not hard to see that the unmatched
parentheses encode the positions of the suffix minima in the sense that if the
unmatched parentheses (except the opening one) 
are the $i_1, i_2 \ldots, i_k$-th opening parentheses in the
current output sequence then the positions $i_1 - 1, \ldots, i_k - 1$ are precisely the
suffix minima positions.  Our task is therefore to sequentially access the next
unmatched parenthesis, starting from the end, when adding the new element $A[i+1]$.
We conceptually break the string into blocks of size $\lfloor \sqrt{n} \rfloor$.
For each block 
that contains at least one unmatched parenthesis, store the following information:

\begin{itemize}
\item its block number (in the original parenthesis string) and the total number of 
open parenthesis in the current output string before the start of the block.
\item the position $p$ of the rightmost parenthesis in the block, and the number 
of open parentheses before it in the block.
\item a pointer to the next block with at least one unmatched parenthesis.
\end{itemize}
This takes $O(\log n)$ bits per block, which is $O(\sqrt{n} \log n)$ bits. 
\begin{itemize}
\item For the rightmost block (in which we add the new parentheses), 
keep positions of all the unmatched parentheses: the space for this is also $O(\sqrt{n} \log n)$ bits.
\end{itemize}
When we process the next element of $A$, we compare it with 
unmatched parentheses in the rightmost block, which takes
$O(1)$ time per unmatched parenthesis that we compared the new element with,  
as in the algorithm of \cite{gbt-stoc84}.
Updating the last block is also trivial.
Suppose we have compared $A[i]$ and found it smaller than all suffix maxima
in the rightmost block.  Then, using the linked list, we find the rightmost unmatched
parenthesis (say at position $p$) in the next block in the list, 
which takes $O(1)$ time, and compare with it 
(this is also $O(1)$ time). If $A[i]$ is smaller, then sequentially scan this block
leftwards starting at position $p$, skipping over a maximal BP sequence to find the next
unmatched parenthesis in that block.   The time for this sequential scan is $O(n)$ 
overall, since we never sequentially scan the same parenthesis twice.  Updating the
blocks is straightforward.  Thus, the creation of the output string can be done
in linear time using $O(\sqrt{n} \log n)$ bits of working memory. 
\end{proof}

\section*{Acknowledgements}
S. R. Satti's research was partly supported by Basic Science Research Program
through the National Research Foundation of Korea funded by the Ministry of
Education, Science and Technology (Grant number 2012-0008241).
\\
P. Davoodi's research was supported by NSF grant CCF-1018370 and BSF grant 2010437 (this work was partially done while P.Davoodi was with MADALGO, Center for Massive Data Algorithmics, a Center of the Danish
National Research Foundation, grant DNRF84, Aarhus University, Denmark).
\bibliographystyle{abbrv}


\end{document}